\documentclass[11pt]{article}

\usepackage[letterpaper,margin=1in]{geometry}
\usepackage[T1]{fontenc}
\usepackage{lmodern,microtype}
\usepackage{amsmath,amssymb,amsthm,mathtools}
\usepackage{aliascnt}
\usepackage{booktabs,tabularx,array}
\usepackage{enumitem}
\usepackage{algorithm,algpseudocode}
\usepackage{xcolor}
\usepackage{xspace}

\usepackage{needspace}
\usepackage[colorlinks=true,linkcolor=blue!45!black,citecolor=blue!45!black,
            urlcolor=blue!45!black]{hyperref}
\usepackage[capitalize,noabbrev,nameinlink]{cleveref}
\usepackage{fancyhdr}

\usepackage{tikz}
\usetikzlibrary{arrows.meta,decorations.pathreplacing,calc}

\usepackage[leftmargin=1em, rightmargin=1em]{quoting}

\newcommand{\MSP}{\ensuremath{\mathsf{MSP}}\xspace}

\setlist{topsep=5pt,itemsep=3pt,parsep=0pt}
\allowdisplaybreaks[1]
\newtheorem{theorem}{Theorem}[section]
\newaliascnt{lemma}{theorem}
\newtheorem{lemma}[lemma]{Lemma}
\aliascntresetthe{lemma}
\crefname{lemma}{Lemma}{Lemmas}
\newaliascnt{corollary}{theorem}

\aliascntresetthe{corollary}
\crefname{corollary}{Corollary}{Corollaries}
\newaliascnt{obs}{theorem}
\newtheorem{obs}[obs]{Observation}
\aliascntresetthe{obs}
\crefname{obs}{Observation}{Observation}
\theoremstyle{remark}
\newaliascnt{example}{theorem}

\aliascntresetthe{example}
\crefname{example}{Example}{Examples}

\usepackage{thm-restate}

\newcommand{\E}{\mathbb E}

\newcommand{\calI}{\mathcal I}

\newcommand{\calM}{\mathcal M}
\newcommand{\config}{W}
\newcommand{\final}{F}
\newcommand{\OPT}{\operatorname{Greedy}}
\newcommand{\ALG}{\operatorname{ALG}}

\title{The Matroid Secretary Conjecture is True}
\author{Sahil Singla\footnote{School of Computer Science, Georgia Institute of Technology, Atlanta, GA, USA. Email: ssingla@gatech.edu. Supported in part by NSF awards CCF-2327010 and CCF-2440113.}}
\date{\today}

\begin{document}
\maketitle
\begin{abstract}
We resolve the matroid secretary conjecture, giving an online algorithm that accepts each element of the offline optimum with probability at least $1/4$. 
The algorithm only needs the number of elements in advance
and independence-oracle access to subsets of already-arrived
elements; it does not need to know the matroid upfront.
\end{abstract}

\section{Introduction}

The classic secretary problem~\cite{Dynkin1963,Lindley1961}
asks us to select the largest of $n$ distinct positive values
revealed in uniformly random order, with immediate and
irrevocable decisions.
A simple algorithm discards the first $n/2$ values and then
selects the first value exceeding all previously observed values.
It succeeds whenever the highest value arrives in the second
half and the second-highest in the first half, an event with
probability at least $1/4$. Discarding the first $n/e$ values instead gives the asymptotically optimal success probability of $1/e$.

Motivated by applications in auction design, Hajiaghayi, Kleinberg, and Parkes \cite{HKP-EC04} and Kleinberg \cite{Kleinberg-SODA05}  extended the secretary problem to selecting up to $r \geq 1$ values with the objective of maximizing the sum of accepted values. Subsequently, Babaioff, Immorlica, and Kleinberg \cite{BIK-SODA07} (the journal version along with Kempe  \cite{BIKK-JACM18}) generalized the problem to selecting multiple values that form an independent set of a matroid  $\calM = (V, \calI)$ with $n=|V|$ elements. This is known as the matroid secretary problem (\MSP); selecting up to $r$ values corresponds to the special case of uniform matroid of rank $r$.
The \emph{matroid secretary conjecture} says that every matroid
admits an $O(1)$-competitive algorithm for \MSP, with a universal
constant independent of the matroid. An online algorithm is
$c$-competitive if its expected total value is at least $1/c$
of the offline optimum.

Matroids form natural feasibility constraints for multiple-value
selection since, in the offline setting, the simple greedy
algorithm, which considers elements in decreasing order of
values and accepts an element whenever feasible, is optimal. Furthermore, greedy is optimal for every assignment of nonnegative values to a downward-closed feasibility system if and only if that system is a matroid~\cite{Edmonds1971}.  
%


The \MSP problem has  been extensively studied in theoretical computer science, both because of its applications in mechanism design  \cite{BabaioffSurvey,Parkes-AGT07} and because it gives a natural model for beyond worst-case analysis of online algorithms \cite{GS-Book21}.
Despite significant progress for special classes of matroids (see \Cref{tab:msp-special-cases-compact}), obtaining an $O(1)$-competitive algorithm for general matroids has remained elusive.  
The $O(\log r)$ guarantee of  \cite{BIK-SODA07}, where $r$ is the rank of matroid, was  improved to $O(\sqrt{\log r})$ by Chakraborty and Lachish \cite{CL-SODA12}. This bound was further improved to $O(\log\log r)$  by Lachish  \cite{Lachish-FOCS14}. Subsequently, Feldman, Svensson, and Zenklusen \cite{FSZ-SODA15}  gave a simpler proof of this $O(\log\log r)$ bound. 
Although $O(1)$-competitive algorithms were known for related
models, such as free-order
secretary~\cite{JSZ-IPCO13,AKW-SODA14},
randomly assigned weights~\cite{Soto-SICOMP13,GV-Journal13},
matroid prophet inequalities~\cite{KW-STOC12}, and
matroid prophet secretary~\cite{EHKS-SICOMP24}, the original \MSP has stood the test of time: the $O(\log\log r)$ bound remained the best general guarantee for more than a decade.

We resolve the matroid secretary conjecture. In fact, our algorithm guarantees that each element of the  optimum is accepted with probability at least $1/4$. Thus, it is even $4$-probability-competitive in the sense of \cite{SotoTV-SODA18}. 

\begin{theorem}[Informal \Cref{thm:matroid-secretary}]
    There exists an online algorithm for the matroid secretary problem that accepts each element of the (canonical) optimum with probability at least $1/4$. The algorithm only needs the size  $n$ of matroid and independence oracle access to the  arrived elements. 
\end{theorem}

By linearity of expectation, this immediately implies a $4$-competitive algorithm for \MSP. 
Moreover, it is ordinal: it only needs comparisons between
observed values, rather than their  magnitudes.

\medskip
\noindent\textbf{Proof idea.} We first study the random-order two-sided Game-of-Googol setting:
each element has two fixed weights, one revealed upfront as a
uniformly random sample and the other revealed online.
A simple simulation reduces matroid secretary to this setting
by pairing each value with zero.
The matroid exchange property tells us that changing one weight
can affect at most one other element of the optimum.
This suggests that random arrival order should give an element
a fair chance of arriving before its only possible competitor.
The difficulty is that this competitor depends on the algorithm's
earlier decisions, so random arrival order alone does not justify
such a symmetry.
Our key idea is to maintain a \emph{reversibly updated sample} (Snapshot \Cref{lem:unique}).
For every fixed arrival order, the updates simply rearrange
the possible samples without biasing them.
The current sample therefore remains independent of the arrival
order, restoring the two-element symmetry despite the algorithm's
adaptivity. 
\Cref{alg:matroidSec} gives an equivalent formulation directly
in the secretary model, and \Cref{sec:lookingBack} explains how
it differs from natural greedy algorithms.

\medskip
\noindent\textbf{Submodular objectives.}
Secretary problems have also been studied for submodular objectives, beginning with
\cite{BHZ-TALG13,GRST-WINE10}, followed by improved guarantees
for several constraints
\cite{FNS-APPROX11,MTW-TOCS16,KT-APPROX17}.
Feldman and Zenklusen~\cite{FZ-SICOMP18} showed that an
$O(1)$-competitive algorithm for \MSP yields one for nonnegative
submodular objectives under the same matroid constraint. Thus, combining \Cref{thm:matroid-secretary} with the reduction of
\cite{FZ-SICOMP18} gives an
$O(1)$-competitive algorithm for the submodular matroid
secretary problem. This holds for every nonnegative
submodular objective, without requiring monotonicity.

\medskip
\noindent\textbf{Intersection of matroids.}
Secretary problems have also been studied under the intersection
of several matroid constraints. Feldman, Svensson, and
Zenklusen~\cite{FSZ-SICOMP22} developed a framework for combining
algorithms for individual matroids, including a preprocessing
step that makes their offline optima overlap. Combining this
preprocessing with an extension of our matroid secretary
algorithm gives an $O(k^3)$-competitive algorithm under the
intersection of $k$ arbitrary matroids; see
\Cref{thm:intersection}. As in the single-matroid case, the
algorithm needs only the number of elements and independence
queries on already-arrived elements in each matroid.

\medskip
\noindent\textbf{Other related work.}
Beyond algorithms for particular matroid classes, prior work
has established connections to contention
resolution~\cite{Dughmi-SICOMP25} and identified limitations
of natural greedy and partition-based
approaches~\cite{BBSW-WINE21,AKKG-ITCS23};
our algorithm uses greedy optimization, but it does not accept every feasible improving element. 
 There is also work on accepting multiple values while satisfying an arbitrary downward-closed feasibility constraint. Here, constant-competitive algorithms are
impossible~\cite{BIKK-JACM18}.
An $O(\log n)$-competitive algorithm is known, matching the lower bound up to a $\mathrm{poly}(\log\log n)$
factor~\cite{RS-STOC26,Rubinstein-STOC16}.
The $O(\log n)$ guarantee of~\cite{RS-STOC26} also applies to XOS objectives. For monotone subadditive objectives, polylogarithmic-competitive algorithms are known~\cite{RS-SODA17}.

\medskip
\noindent\textbf{Concurrent and subsequent work.}
A preliminary version of this paper also announced an
$8$-competitive algorithm for single-sample matroid prophet
inequalities. We omit this application here, as independent
concurrent work of \cite{ABHM-arXiv26} establishes the optimal
$2$-competitive guarantee for every fixed arrival order
independent of the samples and realized values. 
They also obtain a $64$-competitive matroid secretary
algorithm and an $e$-probability-competitive algorithm for linear
matroids in the known-matroid model; the latter result was
independently also obtained in \cite{BDLSV-arXiv26}.
For general matroids, Chan, Lin, and Wang~\cite{CLW-arXiv26}
improved our guarantees to approximately $3.1462$. Subsequently,
Huang~\cite{Huang-arXiv26} resolved the strong matroid secretary
conjecture, obtaining an $e$-probability-competitive algorithm
for every matroid in the  arrived-elements-only oracle model.

\begin{table}[h]
\centering
\small
\setlength{\tabcolsep}{4pt}
\renewcommand{\arraystretch}{1.13}
\begin{tabularx}{\linewidth}{@{}>{\raggedright\arraybackslash}p{0.28\linewidth}
>{\raggedright\arraybackslash}X@{}}
\toprule
\textbf{Matroid class} & \textbf{Guarantees and references}\\
\midrule
Rank one / partition & $e$\,\cite{Lindley1961,Dynkin1963,BIKK-APPROX07}. \\
\addlinespace[3pt]
Uniform (rank $r$) & $O(1)$\,\cite{HKP-EC04}; $1+O(r^{-1/2})$\,\cite{Kleinberg-SODA05}; $[1+O(\sqrt{\log r/r})]^{\mathrm P}$\,\cite{SotoTV-SODA18}; finite-rank refinements\,\cite{BJS-MOR14,CCJ-SODA15,AL-TCS21}. \\
\addlinespace[3pt]
Truncated partition & $O(1)$\,\cite{BIK-SODA07,BIKK-JACM18}; also covered by laminar results. \\
\addlinespace[3pt]
Transversal & $16$\,\cite{DP-ICALP08}\ $\to8$\,\cite{KP-ICALP09}\ $\to e$\,\cite{KRTV-ESA13}; $e^{\mathrm P}$\,\cite{SotoTV-SODA18}. \\
\addlinespace[3pt]
Laminar & $O(1)$\,\cite{IW-SODA11}; $1/0.053$\,\cite{HP-Laminar13}; $3\sqrt3e$\,\cite{JSZ-IPCO13}; $9.6$\,\cite{MTW-TOCS16}; $(3\sqrt3)^{\mathrm P}$\,\cite{SotoTV-SODA18}; $4.75$\,\cite{HPZ-ESA24}; $[1/(1-\ln2)]^{\mathrm P}$\,\cite{BLSV-IPCO25}. \\
\addlinespace[3pt]
Graphic & $16$\,\cite{BIK-SODA07}; $3e$\,\cite{BDGIT-SODA09}; $2e$\,\cite{KP-ICALP09}; $4^{\mathrm P}$\,\cite{SotoTV-SODA18}; $(1/0.2504)^{\mathrm P}$\,\cite{BLSV-IPCO25}; $3.95$\,\cite{BHKKMO-ESA25}. \\
\addlinespace[3pt]
Simple / high-girth graphic & Simple: $3.77$\,\cite{BHKKMO-ESA25}, $(1/0.2693)^{\mathrm P}$\,\cite{BLSV-IPCO25}; girth $g$: $e+o_g(1)$\,\cite{BHKKMO-ESA25}. \\
\addlinespace[3pt]
Cographic; regular / MFMC & $3e$\,\cite{Soto-SICOMP13}; $9e$ for regular and max-flow min-cut\,\cite{DK-SICOMP14}. \\
\addlinespace[3pt]
Hypergraphic; matching & $4^{\mathrm P}$ for both\,\cite{SotoTV-SODA18}; matching matroids have a vertex ground set. \\
\addlinespace[3pt]
$k$-column-sparse; $k$-framed & $ke$ for column-sparse\,\cite{Soto-SICOMP13}; $[k^{k/(k-1)}]^{\mathrm P}$ for both ($k\ge2$)\,\cite{SotoTV-SODA18}. \\
\addlinespace[3pt]
Exchange-structured gammoids / packings & $[\mu^{\mu/(\mu-1)}]^{\mathrm P}$ ($\mu\ge2$); semiplanar gammoids: $(4^{4/3})^{\mathrm P}$\,\cite{SotoTV-SODA18}. \\
\addlinespace[3pt]
Rank two & $(1/0.3462)^{\mathrm P}$\,\cite{BLSV-IPCO25}; all rank-two matroids. \\
\addlinespace[3pt]
Low density / small cocircuits & 
$\gamma(\calM)$ for density $\gamma(\calM)$;
$c$ if every nonloop belongs to a cocircuit of size at most $c$
\,\cite{Soto-SICOMP13}. \\
\addlinespace[3pt]
Paving; almost all matroids & $1+O(r^{-1/2})$ for paving; $2+o(1)$ for almost all labelled matroids\,\cite{HN-SIDMA20}. \\
\addlinespace[3pt]
Proper minor-closed over $\mathbb F_p$ & $O_{\mathcal C}(1)$ for fixed prime $p$, \emph{conditional on structural Hypothesis~1}\,\cite{HN-SIDMA20}. \\
\addlinespace[3pt]
$k$-fold unions &  $1+O(\sqrt{\log n/k})$ for $k\gg\log n$\,\cite{GHKQ-kFold25}. \\
\addlinespace[3pt]
Graphic, arrived-elements-only oracle & $36$ \,\cite{DPLPP-Graphic26}. \\
\addlinespace[3pt]
\bottomrule
\end{tabularx}
\caption{Selected  milestones for special classes of matroids. 
$\mathrm P$ indicates a per-element probability guarantee.
Some results assume that the matroid or a suitable
representation is known upfront.}
\label{tab:msp-special-cases-compact}
\end{table}

\section{Random Order Two-Sided Game-of-Googol}\label{sec:GoG}

Fix a matroid $\calM = (V, \calI)$ with $n=|V|$ elements.
We study the following (random order) matroid version of the two-sided
Game-of-Googol~\cite{CCES-SODA20}.

\medskip
\noindent\textbf{RO $2$-sided GoG:}
Each  element $e \in V$ of the matroid $\calM = (V, \calI)$ has two fixed unknown weights $0\leq \ell_e  < u_e$,  which we call its lower and upper weights respectively. 
At the beginning, each element $e$ independently reveals one of its two weights $S_e \in \{\ell_e,u_e\}$ chosen uniformly at random; we call this its \emph{sample} weight. Let $R_e$ denote the other hidden weight, which we call its \emph{real} weight. The real weights are then revealed sequentially  in a uniformly random order (independent of the sample). When $R_e$ is revealed, the  algorithm must immediately accept or reject  $e$. The accepted elements must form an independent set in $\calM$, and the objective is to maximize their total real weight.

\subsection{Notation and Preliminaries} A \emph{configuration} is a weight vector $w=\big(w_e\big)_{e\in V}$ for $w_e \in \mathbb{R}_{\geq 0}$. 
When the configuration is random, we will instead use $\config=\big(\config_e\big)_{e\in V}$. Write $w^{e\leftarrow v}$ (and similarly $\config^{e\leftarrow v}$) for the configuration obtained by
replacing coordinate $e$ by weight $v \in \mathbb{R}_{\geq 0}$.

Since element weights might have ties or equal zero, we formally define the greedy solution.

\begin{algorithm}
\caption{$\OPT(w)$}
    \begin{itemize}[leftmargin=*]
        \item Consider only elements $e$ with strictly positive weights, i.e. $w_e>0$, in decreasing order of their weights; use the same fixed label order to break ties in
every configuration.
    \item Starting from the empty set, select each element $e$
in this order if adding it to the selected set preserves
independence; equivalently, if $e$ is not spanned by the
preceding elements.
    \item Return the set of selected elements.
    \end{itemize}
\end{algorithm}


\noindent  We write $\OPT(w)$ for the returned set.
It is a maximum-weight basis of the restriction to 
$\{e:w_e>0\}$, but need not be a basis of the whole matroid.  


We will  need the well-known property that an element in the optimum continues to be in the optimum even after some other weights are reduced.

\begin{lemma}[Greedy monotonicity] \label{lem:greedy}
Consider any two configurations $w,w'$ such that $w'_e \leq w_e$ for every element $e\in V$. If  $e^\star \in \OPT(w)$ and $w_{e^\star} = w'_{e^\star}$  then $e^\star \in \OPT(w')$.
\end{lemma}
\begin{proof}
The set $P'$ of elements preceding $e^\star$ under the greedy order for $w'$ is contained in the
corresponding set $P$ under $w$, since weights only decrease while
$w'_{e^\star}=w_{e^\star}$ and ties are broken by labels.
Greedy selects $e^\star$ under $w$, so
$e^\star\notin\operatorname{Span}(P)$.
Since $P'\subseteq P$, we also have
$e^\star\notin\operatorname{Span}(P')$, and hence greedy selects
$e^\star$ under $w'$.
\end{proof}


We also use the following  consequence of Brualdi's 
symmetric basis-exchange theorem~\cite{Brualdi1969};
see also~\cite[Theorem~39.12]{Schrijver2003}.
We include a proof for completeness, accounting for our
convention of omitting zero-weight elements. This lemma will form the main intuition behind our algorithm.

 \begin{lemma}[One-coordinate sensitivity]\label{lem:heart}
 If two configurations $w,w'$ differ only at element $e$,
then their greedy optima are either identical, or differ only in the
membership of $e$, or differ by an exchange of $e$ with one
other element $f$. 
\end{lemma}

\begin{proof}
If $w_e=w'_e$, the statement is immediate. Otherwise, assume
$w'_e<w_e$, swapping the roles of $w$ and $w'$ if necessary. The two greedy
scans use the same relative order on $V\setminus\{e\}$, and differ only in
where $e$ is inserted: earlier under $w$, later (or not at all, if $w'_e=0$)
under $w'$. Let $P$ be the set of elements scanned before $e$ in the $w$-order;
both runs select the same set $A$ on $P$.

If $e\in\operatorname{Span}(A)$, the $w$-run rejects $e$.
The $w'$-run either skips $e$ because $w'_e=0$, or reaches it
with a superset of $A$ and also rejects it.
Since $e$ is selected by neither run, all other decisions coincide, and $\OPT(w')=\OPT(w)$.

Otherwise the $w$-run selects $e$, and 
we compare the two runs on the common list of remaining
elements other than $e$
from the sets $A\cup\{e\}$ and $A$. We claim their decisions agree until
at most one element $f$, which is accepted by the $A$-run and rejected by the
$A\cup\{e\}$-run, after which the two current sets have equal span and all later decisions on elements other than $e$ coincide.
Indeed, spans are monotone, so the first disagreement
must have this form; and after it the two sets are $A\cup\{e\}\cup B$ and
$A\cup B\cup\{f\}$, which are independent, of equal size, and the first spans the second, so their spans agree. Here $B$ consists of the elements accepted by both runs after $P$ and before $f$.

First suppose $w'_e>0$.
If this $f$ occurs before $e$'s position in the $w'$-order, then the $w'$-run
reaches $e$ with a set spanning $e$ and rejects it, giving
$\OPT(w)=(\OPT(w')\setminus\{f\})\cup\{e\}$. If no disagreement occurs before that point, then $A\cup B\cup\{e\}$ is independent, so the $w'$-run also selects $e$ and the two runs coincide from then on, giving $\OPT(w')=\OPT(w)$. If
$w'_e=0$, the $w'$-run never scans $e$ and the same dichotomy yields
$\OPT(w)=\OPT(w')\cup\{e\}$ or $\OPT(w)=(\OPT(w')\setminus\{f\})\cup\{e\}$.
\end{proof}

\subsection{Algorithm}

The algorithm is inspired by \Cref{lem:heart}. We maintain a configuration $\config$, which is initialized to sample weights $S$. Before an element $e$ is processed,  coordinate $\config_e$ still equals $S_e$, since only the iteration that considers $e$ can change it.
At arrival of $i$-th element $e$, we compute the  greedy solutions of $\config$, both with  $e$ taking sample weight $S_e$ (i.e., $\OPT(\config)$) and with $e$ taking real weight $R_e$  (i.e., $\OPT(\config^{e\leftarrow R_e})$). 

By \Cref{lem:heart}, changing the weight of $e$ can affect the
membership of at most one other element. We proceed only if
the two greedy sets agree on the elements that arrived before $e$ (the preceding elements are  in  \emph{finalized} set $\final$). 
If this test passes, we accept $e$ exactly when it belongs to
$\OPT(\config^{e\leftarrow R_e})$, and update $\config$ whether
or not $e$ is accepted. If the test fails, we reject $e$ and leave
$\config$ unchanged. In both cases, $e$  is finalized by adding it to $\final$.

Observe that a permitted update can even reduce the weight of $e$ in $\config$: our analysis crucially exploits this symmetry between $S_e$ and $R_e$.

\begin{algorithm} 
\caption{RO $2$-Sided GoG}\label{alg:SSProphSec}
    \begin{itemize}[leftmargin=*]
    \item Let $S_e$ and $R_e$ denote the random sample and real weights of each element $e$, respectively. 
        \item Initialize accepted set $\ALG = \emptyset$, configuration vector $\config = S$, and finalized set $\final = \emptyset$.
    
    \item For $i=1$ to $n$:
\begin{enumerate}
    \item Let $e$ denote the $i$-th element in random order $\Pi$. \label{alg:step1}
    \item If $ \OPT(\config^{e\leftarrow R_e})\cap \final = \OPT(\config )\cap \final $   \Comment{Both have same  finalized elements}\label{alg:step2}
        \begin{enumerate}
            \item If $e\in \OPT(\config^{e\leftarrow R_e})$ then $\ALG \gets \ALG \cup \{e\}$ \Comment{accept $e$}  \label{alg:step2a}
            \item $\config \gets \config^{e\leftarrow R_e}$   \Comment{Update $W$, even if $e$ is not accepted} \label{alg:step2b}
        \end{enumerate}
    \item $\final \gets \final\cup \{e\}$.    \Comment{$\final$ contains all processed/finalized elements}
\end{enumerate}
    \item Return $\ALG$.
\end{itemize}
\end{algorithm}


\begin{theorem}\label{thm:ROGoG}
    \Cref{alg:SSProphSec} always  accepts an independent set and guarantees for each  $e^\star \in \OPT\big((u_e)_{e\in V}\big)$ that  
    \[\Pr[\ALG \ni e^\star \text{ and } R_{e^\star}=u_{e^\star}] ~\geq~ 1/4. \]
\end{theorem}
\begin{proof}
    For feasibility, we maintain the invariant that $\ALG \subseteq \OPT(\config)$, which is always feasible. It holds initially. 
    Every previously accepted element $e$ belongs to $F$, and every subsequent update ensures that $e$ belongs to $\OPT(\config)$ due to Step~\ref{alg:step2}. Moreover, Step~\ref{alg:step2a}  ensures that the newly accepted element belongs to the proposed optimal set. Thus the invariant holds after every iteration.
    
    Fix $e^\star\in\OPT\big((u_e)_{e\in V}\big)$. For  the rest of the proof,  condition on $R_{e^{\star}}=u_{e^\star}$ (and  hence  $S_{e^\star}=\ell_{e^\star}$), which happens with probability $1/2$. It suffices to prove that
    \[ \Pr[\ALG \ni e^\star \mid R_{e^\star}=u_{e^\star}] ~\geq~ 1/2. \]

    Let $\config^\star$ denote  the (random) configuration $\config$ in Step~\ref{alg:step1} of  \Cref{alg:SSProphSec} when it encounters $e^\star$. Observe that $\config^\star_{e^\star} = \ell_{e^\star}$ since we have conditioned on $S_{e^\star}=\ell_{e^\star}$ and no earlier element changes coordinate $e^\star$ of $\config$. 
    Moreover, $e^\star$ will pass Step~\ref{alg:step2a} as 
    $e^\star \in \OPT((\config^\star)^{e^\star\leftarrow u_{e^\star}})$  by \Cref{lem:greedy}. Hence, $e^\star $ will be accepted iff it passes Step~\ref{alg:step2} of the algorithm. We will prove that this occurs with probability at least $1/2$.
    
    Condition further on $\config^\star = w$.  Moreover, $\config^\star_{e^\star}=S_{e^\star}$, so conditioning on
$\config^\star=w$ with $w_{e^\star}=\ell_{e^\star}$ already implies that $R_{e^\star}=u_{e^\star}$. We consider two cases:

\smallskip
\noindent\emph{Case 1:}     If $\OPT(w^{e^\star\leftarrow u_{e^\star}})$ is either
$\OPT(w)$ or $\OPT(w)\cup\{e^\star\}$, then the test in Step~\ref{alg:step2}
passes automatically  since $e^\star\notin F$, and $e^\star$ is accepted. 

\smallskip
\noindent\emph{Case 2}: Otherwise, there is a unique exchange partner $f^\star$
as given by \Cref{lem:heart} such that $\OPT(w^{e^{\star}\leftarrow u_{e^\star}}) = (\OPT(w) \setminus \{f^\star\})\cup \{e^\star\}$. 
    Since $e^\star\notin F$, we have $e^\star$ passes Step~\ref{alg:step2} iff $f^\star$ is not in $F$, i.e., the real weight
of $f^\star$ has not yet been revealed (in other words, $e^\star$ comes before $f^\star$ in arrival order $\Pi$).
For fixed $w$, the exchange partner $f^\star$ is fixed, and the following
\Cref{lem:unifPerm} shows that the arrival order remains uniform
after conditioning on $\config^\star=w$.
    Hence,
    \[
    \Pr[\ALG \ni e^\star \mid \config^\star =w]  \quad = \quad  \Pr[ e^\star \text{ before } f^\star \mid  \config^\star =w] \quad \overset{\text{\Cref{lem:unifPerm}}}{=} \quad 1/2.
    \]
    Hence,     we have that
     \begin{align*}
         \Pr[\ALG \ni e^\star \mid R_{e^\star}=u_{e^\star}]  &= \sum_{w: w_{e^\star}=\ell_{e^\star}} \Pr[\config^\star =w \mid R_{e^\star}=u_{e^\star}] \cdot \Pr[\ALG \ni e^\star \mid R_{e^\star}=u_{e^\star}, \config^\star =w] \\
         &\geq  \sum_{w: w_{e^\star}=\ell_{e^\star}} \Pr[\config^\star =w \mid R_{e^\star}=u_{e^\star}] \cdot 1/2  \quad = \quad 1/2,
    \end{align*}
    where the inequality uses that the conditional probability is  $1$ in the first case and exactly  $1/2$ in the exchange case.  Multiplying by
$\Pr[R_{e^\star}=u_{e^\star}]=1/2$
completes the proof.
\end{proof}

In the following two lemmas, fix an element $e^\star$, and let
$\config^\star$ denote the configuration immediately before its
real weight is revealed. The matroid and all weight pairs are
fixed, and all probabilities refer to the original,
unconditioned experiment.

We first prove the \emph{snapshot lemma} (\Cref{lem:unique}):
given the arrival order and a ``snapshot'' of the configuration
faced by $e^\star$, we can uniquely reconstruct the initial
sample, and hence the entire trajectory of the algorithm,
including its states and decisions. This reconstruction is
only for the analysis. 

    \begin{lemma}[Snapshot lemma]\label{lem:unique}
For every $w\in \prod_{e\in V}\{\ell_e,u_e\}$ and every arrival order $\pi$, there exists
a unique sample vector $S_{w,\pi}$ that makes $e^\star$ encounter
configuration $\config^\star=w$.
Consequently, the entire trajectory is uniquely determined
by $w$ and $\pi$.
    \end{lemma}
    
    \begin{proof}
    We will prove that there is a unique sample $S_{w,\pi}$ that resulted in $\config^\star=w$ for arrival order $\pi$. Given  this $S_{w,\pi}$ (along with order $\pi$), the trajectory of the algorithm is deterministic.
    
    Let $k$ denote the position of $e^\star$ in order $\pi$, i.e., $e^\star = \pi(k)$.
    Firstly, observe that all elements at positions $t \geq k$ in $\pi$ have not been processed when $e^\star$ arrives, so  $S_{\pi(t)} = w_{\pi(t)}$.

    Next, we use backward induction (reverse order from $k-1$ to $1$) to recover the sample vector $S_{w,\pi}$. We work backward because the test at position $t$ uses the sample weights of later elements, which may differ from their
weights in the snapshot and must therefore be recovered first.
Inductively, suppose $S_{\pi(t+1)},\ldots,S_{\pi(n)}$ have
been uniquely recovered; the preceding observation establishes
the base case $t=k-1$. Then 
    we know every coordinate of $W$ immediately before processing
$\pi(t)$, except possibly coordinate $\pi(t)$:
coordinates at later positions still have their reconstructed
sample weights, while coordinates at earlier positions have
their weights in $w$ because they are never updated again. Thus, the only thing unclear is whether the sample contains the lower or upper weight at position $t$. 
    
    The crucial observation is that irrespective of which weight is in the sample at position $t$, the test in Step~\ref{alg:step2} of the algorithm returns the same answer. Indeed, the finalized set is $F=\{\pi(1),\ldots,\pi(t-1)\}$, and the two possibilities
compare the same two optimal sets against $F$, with the
sides of the equality interchanged.  If the test passes, then we know that sample and real will be swapped, and in the other case they won't be swapped. In either case, since we know $w_{\pi(t)}$, we can calculate the sample at position $t$ by reversing the step. Each reverse step therefore determines exactly one valid
sample entry. Reversing all steps produces a unique sample
vector, and running the algorithm forward from this vector
gives $\config^\star=w$, proving both existence and uniqueness.
    \end{proof}

    For each fixed arrival order, \Cref{lem:unique} gives a bijection
between initial samples and snapshots. Since all samples are
equally likely, the snapshot is uniform regardless of the
arrival order. This yields the following independence property. 

    \begin{lemma}[Independence of $\config^\star$ and $\Pi$] \label{lem:unifPerm}
    For any configuration $ w \in \prod_{e\in V}\{\ell_e,u_e\}$ and
    any  permutation $\pi$, 
        \[\Pr[\Pi = \pi \mid \config^\star = w] \quad =\quad  \frac{1}{n!} .\]
    \end{lemma}

    \begin{proof}
    Fix a permutation $\pi$. Conditional on $\Pi=\pi$, the sample
vector is uniform over the $2^n$ configurations in $\prod_{e\in V}\{\ell_e,u_e\}$.
By  \Cref{lem:unique}, exactly one of these samples produces
$\config^\star=w$. Hence,
    $\Pr[\config^\star=w\mid\Pi=\pi]=2^{-n}$ for every  $w$ and $\pi$,
which implies 
\[ \Pr[\config^\star=w] \quad = \quad  \sum_{\pi} \Pr[\Pi=\pi] \cdot \Pr[\config^\star=w\mid\Pi=\pi] \quad = \quad  \sum_{\pi} \frac{1}{n!} \cdot 2^{-n} \quad  = \quad  2^{-n}.\]
    Hence,
        \[
        \Pr[\Pi = \pi \mid \config^\star = w]\quad  = \quad \Pr[\Pi = \pi] \cdot  \frac{\Pr[\config^\star = w \mid \Pi = \pi]  }{\Pr[\config^\star=w]} \quad = \quad \frac{1}{n!} \cdot \frac{2^{-n}}{2^{-n}} \quad = \quad \frac{1}{n!}.     \qedhere
        \]
     \end{proof}

\section{Matroid Secretary}

We reduce matroid secretary to RO $2$-sided GoG and apply
\Cref{thm:ROGoG}. The  idea of the reduction is to pair each secretary value with zero value, generate the GoG samples by discarding a random prefix of \MSP, and simulate the zero-valued arrivals by interleaving discarded elements with the remaining matroid secretary arrivals.

Fix an unknown matroid $\calM = (V, \calI)$ with $n=|V|$ elements and 
unknown nonnegative element values $v_e\geq0$  for $e\in V$. 
The number of elements $n$ is known in advance, and ties are
broken using fixed element labels, not their arrival positions.
We call the secretary inputs \emph{values}, to distinguish
them from the weights in the virtual GoG instance.
The values are fixed before a uniformly random  arrival
order $\pi'$ is drawn.
We may assume $v_e>0$ for every $e$: replacing $v$ by $v+\mathbf 1$ leaves the greedy order of $\{e : v_e>0\}$ unchanged and merely appends the remaining elements at the end, so $\OPT(v)\subseteq\OPT(v+\mathbf 1)$ and a per-element
guarantee for $\OPT(v+\mathbf 1)$ implies one for $\OPT(v)$.

\begin{restatable}{theorem}{matroidSec}\label{thm:matroid-secretary}
    There exists a randomized online algorithm such that for every instance of the matroid secretary problem with fixed, initially unknown non-negative values $(v_e)_{e\in V}$ and known $n=|V|$, it always accepts an independent set and accepts each element  $e^\star \in \OPT((v_e)_{e\in V})$ with probability at least $1/4$. The algorithm does not need to know the  matroid upfront: it only needs independence oracle access to the already arrived elements. 
\end{restatable}

\begin{proof}
For any instance of \MSP with (unknown) element values $(v_e)_{e\in V}$, our reduction defines a \emph{virtual} instance of RO $2$-sided GoG with $\ell_e =0$ and $u_e=v_e$ for all $e \in V$. Next, we discuss how to generate a random virtual sample $S$ and a virtual arrival order $\pi$  for this RO $2$-sided GoG instance, using the arrival order $\pi'$ for \MSP and fresh randomness.

\medskip
\noindent\textbf{Generating  virtual sample.} To generate the sample $S$ for $2$-sided GoG, we choose a random $K$ from the binomial distribution $\text{Bin}(n,1/2)$, independently of the arrival order $\pi'$, and then observe and discard the first $K$ elements of \MSP arriving in order $\pi'$. Let $\mathcal{K}$ denote this set of discarded elements, so $K=|\mathcal{K}|$. Now we define the sample and real weights for GoG instance  as follows:
\[
    (S_e,R_e) = (v_e,0) \text{ for } e\in \mathcal{K} \quad \text{and} \quad (S_e,R_e) = (0,v_e) \text{ for } e\not\in \mathcal{K}.
\]
Observe that in the sample $S$, each element $e$ takes weight $\ell_e=0$ or $u_e=v_e$ uniformly and independently, as desired by RO $2$-sided GoG.


\medskip
\noindent\textbf{Generating virtual order.} We generate $\pi$ for the $2$-sided GoG using the random order $\pi'$ of \MSP.
Notice that we cannot set $\pi = \pi'$: although it would still imply that $\pi$ is a uniformly random permutation, virtual order $\pi$ will not be independent of the virtual sample $S$; e.g., all $K$ elements with $S_e = u_e$ will arrive before those with $S_e =0$. 
To obtain the required independence, we generate $\pi$
as follows (also see \Cref{fig:virtual-order}), using fresh independent randomness:
\begin{enumerate}
    \item First, mark a uniformly chosen $K$-element subset
of the $n$ virtual positions, i.e., one of the  $\binom nK$ subsets. 
    \item Place the $K$ discarded elements from \MSP into these $K$ marked positions in the order given by $\pi'$. In other words, the projection of $\pi$ on these marked positions is the same as projection of $\pi'$ on $\mathcal{K}$.

   \item Fill the unmarked positions of $\pi$ with the remaining
secretary elements in their actual arrival order $\pi'$,
as they arrive.
In other words, the projection of $\pi$ on these unmarked positions is the same as projection of $\pi'$ on $V\setminus\mathcal{K}$.
\end{enumerate}

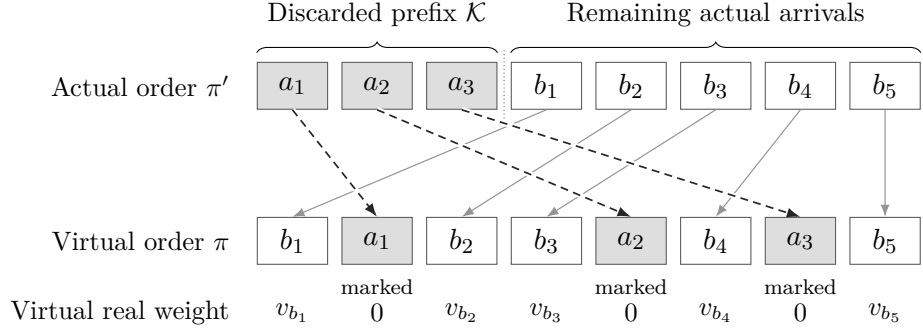
\begin{figure}[t]
    \centering
\begin{tikzpicture}[
    x=1.12cm,y=1cm,
    slot/.style={draw=black!65,fill=white,minimum width=0.93cm,
                 minimum height=0.62cm,inner sep=0pt,font=\normalsize},
    ghost/.style={slot,fill=black!13},
    livearrow/.style={-{Latex[length=1.6mm]},draw=black!42,line width=0.45pt},
    ghostarrow/.style={-{Latex[length=1.8mm]},draw=black!85,
                      line width=0.65pt,densely dashed},
    every node/.style={font=\small}
]
\foreach \i/\lab in {1/a_1,2/a_2,3/a_3}
    \node[ghost] (p\i) at (\i,0) {$\lab$};
\foreach \i/\lab in {4/b_1,5/b_2,6/b_3,7/b_4,8/b_5}
    \node[slot] (p\i) at (\i,0) {$\lab$};
\node[anchor=east,align=right] at (0.38,0)
    {Actual order $\pi'$};
\draw[densely dotted,black!60] (3.5,-0.43)--(3.5,0.48);
\draw[decorate,decoration={brace,amplitude=4pt}]
    ($(p1.north west)+(0,0.12)$)--($(p3.north east)+(0,0.12)$)
    node[midway,above=6pt] {Discarded prefix $\mathcal K$};
\draw[decorate,decoration={brace,amplitude=4pt}]
    ($(p4.north west)+(0,0.12)$)--($(p8.north east)+(0,0.12)$)
    node[midway,above=6pt] {Remaining actual arrivals};
\foreach \i/\lab in {1/b_1,3/b_2,4/b_3,6/b_4,8/b_5}
    \node[slot] (v\i) at (\i,-2.05) {$\lab$};
\foreach \i/\lab in {2/a_1,5/a_2,7/a_3}
    \node[ghost] (v\i) at (\i,-2.05) {$\lab$};
\node[anchor=east] at (0.38,-2.05) {Virtual order $\pi$};
\foreach \from/\to in {4/1,5/3,6/4,7/6,8/8}
    \draw[livearrow] (p\from.south)--(v\to.north);
\foreach \from/\to in {1/2,2/5,3/7}
    \draw[ghostarrow,preaction={draw=white,line width=2.1pt}]
        (p\from.south)--(v\to.north);
\foreach \i in {2,5,7}
    \node[anchor=north,font=\scriptsize] at (\i,-2.43) {marked};
\node[anchor=east] at (0.38,-3.00) {Virtual real weight};
\foreach \i/\wt in {1/v_{b_1},2/0,3/v_{b_2},4/v_{b_3},5/0,6/v_{b_4},7/0,8/v_{b_5}}
    \node at (\i,-3.00) {$\wt$};
\end{tikzpicture}
    \caption{{\small Generating the virtual order $\pi$ from the actual order $\pi'$. The marked positions  in  $\pi$ form a uniformly random $|\mathcal K|$-subset of $\{1,\ldots, n\}$ drawn independently of $\pi'$. 
    The relative orders of both the discarded prefix
    and the remaining actual arrivals are preserved. Shaded elements are   processed virtually with real weight $0$ and unshaded elements are processed  upon their actual arrivals with revealed values.}}
    \label{fig:virtual-order}
\end{figure}

\begin{obs}\label{obs:algSimulate}
While the secretary elements arrive in order $\pi'$, we can
simulate \Cref{alg:SSProphSec} for $2$-sided GoG in virtual order $\pi$ using only independence queries on the actually arrived elements.
\end{obs}

\begin{proof}[Proof of \Cref{obs:algSimulate}]
After observing and rejecting the first $K$ physical arrivals,
we initialize $W=S$ and $F=\emptyset$ in \Cref{alg:SSProphSec}.
At a marked position of $\pi$, we process the corresponding
discarded element with real weight $0$.
At an unmarked position, we observe the next physical arrival
in $\pi'$, process it with its revealed value, and follow the
algorithm's acceptance decision.
We perform every permitted update to $\config$ and add every virtually
processed element to $F$, including elements at marked positions.
Discarded elements are never accepted, since their real weight
is zero. Thus $F$ records the virtual past, not the actual past.
We process any intervening marked positions before observing
the next physical arrival, and decide on that arrival before
inspecting another one.

Every positive-weight element in either the current or proposed
configuration has already arrived: initially these
elements belong to $\mathcal K$, and subsequently a positive
weight can be introduced only when its element  arrives.
Since $\OPT$ runs greedy only on positive-weight elements,
all its independence queries involve  arrived elements.
Unseen coordinates can simply be stored implicitly as zero.
\end{proof}

\paragraph{Analysis.} 
The main lemma is that the virtual order $\pi$ is a uniformly random permutation independent of the virtual sample $S$.

\begin{lemma}\label{lem:indepPermSet}
Let $\Pi$ denote the random virtual permutation generated above.  Then, for every permutation $\pi$ and every sample vector 
$s \in \prod_{e\in V}\{0,v_e\}$, we have
    \[\Pr[\Pi=\pi \mid  S=s] \quad = \quad \frac{1}{n!}.\]
\end{lemma}

\begin{proof} [Proof of \Cref{lem:indepPermSet}]
Since every $v_e>0$, conditioning on $S=s$ identifies
$\mathcal K=\{e:s_e>0\}$ and its size $k$.

Conditioned on $S=s$, the order $\pi'$ is uniform among the $k!(n-k)!$ orders whose first $k$ are exactly $\mathcal{K}$; hence the relative order of $\mathcal{K}$ and $V \setminus \mathcal{K}$ are independent and uniform. 
For any permutation $\pi$ of $V$, obtaining $\Pi=\pi$
requires one particular marked-position set, one particular
order of the discarded elements, and one particular suffix order. Hence, we have
\[
    \Pr[\Pi=\pi \mid S=s]
    \quad =\quad \frac{1}{\binom nk}\frac{1}{k!}\frac{1}{(n-k)!}
    \quad = \quad \frac1{n!}.    \qedhere
\]
\end{proof}

By \Cref{lem:indepPermSet},  the virtual order $\pi$ is independent of virtual sample $S$, and both are uniform. We will execute the  $2$-sided GoG algorithm from \Cref{thm:ROGoG} on virtual order $\pi$, which we know is possible by \Cref{obs:algSimulate}. 
Let $\ALG$ be the output of this algorithm. Let 
$A=\ALG\setminus\mathcal K$ be the matroid secretary output.
Since $\ALG$ is an independent set of the matroid, so is $A$. 
Moreover, $R_e=v_e$ holds exactly when $e\notin\mathcal K$. Fix
$e^\star\in\OPT((v_e)_{e\in V})
=\OPT((u_e)_{e\in V})$. We have
\[
    \Pr[e^\star\in A]
    \quad =\quad 
    \Pr[e^\star\in\ALG \text{ and }
        R_{e^\star}=v_{e^\star}]
    \quad \overset{\text{\Cref{thm:ROGoG}}}{\geq} \quad \frac14.    \qedhere
\]
\end{proof}

\medskip
\noindent \textbf{Direct secretary algorithm.}
\Cref{alg:matroidSec} writes the above simulation directly in
secretary terminology. For $T\subseteq V$, write $\OPT(T)$ for
greedy on $T$ using the values $v_e$ and the same fixed
tie-breaking rule. Choosing $X$ uniformly from $\{0,1\}^n$ is
equivalent to drawing $K\sim\operatorname{Bin}(n,1/2)$ and then
marking a uniformly random $K$-subset of the virtual positions.
Under the positive-value convention above, $C$ is exactly the
set of positive coordinates of the virtual configuration $W$:
a sample step proposes a deletion, and a real step proposes an
insertion. Thus the updates and acceptance decisions coincide
with the simulation, and \Cref{thm:matroid-secretary} applies.
Here $F$ records the virtual past, not all actually observed
elements. The algorithm uses $O(n^2)$ independence queries, since each of
its $n$ virtual steps requires at most two greedy scans on at
most $n$ elements.

\begin{algorithm} 
\caption{Matroid Secretary} \label{alg:matroidSec}
 \begin{itemize}[leftmargin=*]
    \item Draw a uniformly random  vector $X$ in $\{0,1\}^n$, and let $K = \|X\|_1$ be the number of $1$s in $X$.  
    \item Discard the first $K$ arrivals and call them sample $S$. The remaining elements are called reals.  
    \item Initialize accepted set $\ALG = \emptyset$, current set $C =S$, and finalized set $\final = \emptyset$.
    
    \item For $i=1$ to $n$:
\begin{enumerate}
    \item  If $X_i = 1$ then $e$ denotes the next sample element (in the original arrival order), and otherwise $e$ denotes the next real element. Let $C' := C  \Delta \{e\} = (C\setminus\{e\}) \cup (\{e\}\setminus C)$.
    
    \item If $ \OPT(C')\cap \final = \OPT(C )\cap \final $   \Comment{Both have same  finalized elements}
        \begin{enumerate}
            \item If $e\in \OPT(C')$ then $\ALG \gets \ALG \cup \{e\}$ \Comment{accept $e$}   
            \item $C \gets C'$   \Comment{Update $C$, even if $e$ is not accepted}  
        \end{enumerate}
    \item $\final \gets \final\cup \{e\}$.    \Comment{$\final$ contains all virtually processed/finalized elements}
\end{enumerate}
    \item Return $\ALG$.
\end{itemize}
\end{algorithm}

\section{Looking Back}
\label{sec:lookingBack}

Most existing algorithms for matroid secretary exploit Greedy
monotonicity (\Cref{lem:greedy}), i.e., an element
$e^\star\in\OPT(V)$ of the global optimum is also part of the
optimal solution on any subset containing $e^\star$. This suggests
a na\"ive algorithm that discards the first roughly $n/2$
elements to use as a sample $S$ and then accepts an element $e$ if
and only if $e\in\OPT(S\cup\{e\})$. (In GoG, the analogous rule
uses the given sample and tests whether
$e\in\OPT(S^{e\leftarrow R_e})$.) In the secretary setting, this
algorithm has good total value since each $e^\star\notin S$ will
be accepted, but it may not return a feasible independent set.
For example, consider the graphic matroid in
\Cref{fig:looking-back}(a): the sample consists of edges
$(u,a),(u,b),(u,c)$ of weight $1$, and the real elements are
$(a,b),(b,c),(c,a)$ of weight $2$. The na\"ive algorithm accepts
all three real edges, which form a cycle.

\medskip
\noindent\textbf{Greedy algorithms.}
A natural approach to go beyond the na\"ive algorithm maintains
an independent set $I_t$ at all times $t$, initialized to the
optimum solution on the sample, and containing all accepted
elements. It accepts the next real element $e$ if and only if the
optimum solution on $I_t\cup\{e\}$ \emph{after contracting the accepted
elements} contains $e$, and updates $I_t$ to this new solution
(with the contracted elements included).

The authors of \cite{BBSW-WINE21} considered this and related
``greedy'' algorithms. 
More formally, their algorithms maintain an independent set
$I_t$ that spans the previously observed elements and satisfies
two properties: (i)~$I_t$ contains all
accepted elements, and (ii)~$I_t$ contains no real element rejected
after the sampling stage. The choice of $I_t$ is flexible, but the
contracted-greedy acceptance rule above is mandatory. They show
that this class cannot be constant-competitive on graphic
matroids, using hat instances with many claws.

A small hat illustrates the failure mechanism; see
\Cref{fig:looking-back}(b). The sample contains a $u$--$v$ path
with edge weights $4,2$, while the real elements form another
$u$--$v$ path with weights $1,3$, followed by a base edge of
arbitrarily large weight $H$. The greedy algorithm accepts the
weight-$1$ edge since it preserves independence. It then accepts
the weight-$3$ edge: after contracting the first acceptance, the
optimum keeps weights $4,3,1$ and drops the weight-$2$ sample.
The two accepted edges now connect $u$ to $v$, making the heavy
base edge infeasible.

\begin{figure}[htbp]
\centering
\begin{minipage}[t]{0.47\linewidth}
\centering
\begin{tikzpicture}[
    x=1cm,y=1cm,
    vertex/.style={circle,fill=black,inner sep=1.5pt},
    sample/.style={densely dashed,line width=0.8pt},
    real/.style={line width=0.9pt},
    wt/.style={fill=white,inner sep=1.5pt,font=\small},
    every node/.style={font=\small}
]
\coordinate (a) at (0,1.50);
\coordinate (b) at (-1.55,-1.00);
\coordinate (c) at (1.55,-1.00);
\coordinate (u) at (0,-0.12);
\draw[real] (a)--node[wt,pos=0.5] {$2$} (b);
\draw[real] (b)--node[wt,pos=0.5] {$2$} (c);
\draw[real] (c)--node[wt,pos=0.5] {$2$} (a);
\draw[sample] (u)--node[wt,pos=0.56] {$1$} (a);
\draw[sample] (u)--node[wt,pos=0.57] {$1$} (b);
\draw[sample] (u)--node[wt,pos=0.57] {$1$} (c);
\node[vertex,label=above:$a$] at (a) {};
\node[vertex,label=below left:$b$] at (b) {};
\node[vertex,label=below right:$c$] at (c) {};
\node[vertex,label=above left:$u$] at (u) {};
\end{tikzpicture}

\smallskip
{\small (a) The na\"ive rule accepts a cycle.}
\end{minipage}
\hfill
\begin{minipage}[t]{0.49\linewidth}
\centering
\begin{tikzpicture}[
    x=1cm,y=1cm,
    vertex/.style={circle,fill=black,inner sep=1.5pt},
    sample/.style={densely dashed,line width=0.8pt},
    real/.style={line width=0.9pt},
    wt/.style={fill=white,inner sep=1.5pt,font=\small},
    every node/.style={font=\small}
]
\coordinate (u) at (-1.65,-1.00);
\coordinate (v) at (1.65,-1.00);
\coordinate (x) at (0,1.50);
\coordinate (y) at (0,0.20);
\draw[sample] (u)--node[wt,pos=0.5] {$s_1:4$} (x);
\draw[sample] (x)--node[wt,pos=0.5] {$s_2:2$} (v);
\draw[real] (v)--node[wt,pos=0.5] {$e_1:1$} (y);
\draw[real] (y)--node[wt,pos=0.5] {$e_2:3$} (u);
\draw[real,line width=1.2pt] (u)--node[wt] {$h:H$} (v);
\node[vertex,label=below left:$u$] at (u) {};
\node[vertex,label=below right:$v$] at (v) {};
\node[vertex,label=above:$x$] at (x) {};
\node[vertex,label=above:$y$] at (y) {};
\end{tikzpicture}

\smallskip
{\small (b) A hat; real arrival order $e_1,e_2,h$.}
\end{minipage}
\caption{Dashed edges belong to the sample; solid edges are real.
In (a), every real edge individually belongs to the optimum when
added to the sample. In (b), contracting $e_1$ allows greedy to
accept $e_2$, whereas ordinary greedy would drop $e_1$.}
\label{fig:looking-back}
\end{figure}
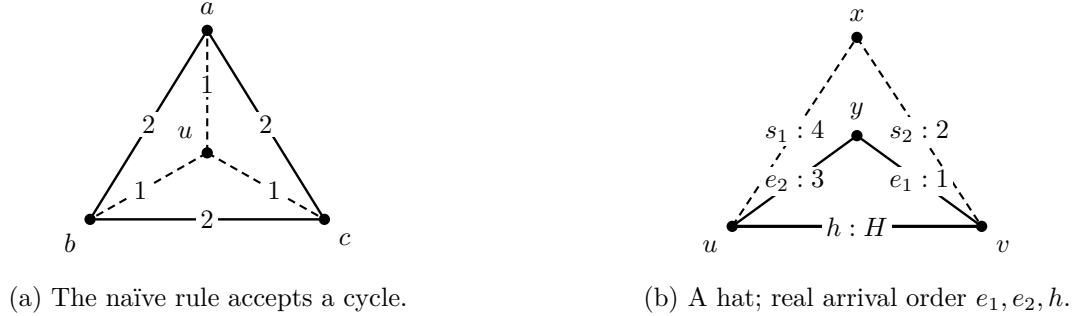

Our Algorithms~\ref{alg:SSProphSec} and~\ref{alg:matroidSec}
build on two crucial ideas.

\medskip
\noindent\textbf{Idea 1: Preserve finalized greedy decisions (Step 2).}
Rather than making the greedy decision from a single independent
set $I_t$ after contracting the accepted elements, maintain a
configuration $w$ (a current set $C$ in the secretary algorithm)
and only consider $e$ if the proposed update preserves greedy
membership on the finalized elements. In particular, we \emph{do not
contract} the already accepted elements; instead, Step~2 ensures
that the previous/finalized memberships do not change. This
immediately guarantees feasibility, since the accepted elements
always belong to the optimum of the current configuration.

On the hat above, after sampling $s_1,s_2$, consider processing order
$e_1,e_2,h$ in our secretary algorithm (and neither sample has been replayed).
The weight-$1$ edge $e_1$
is accepted, but inserting the weight-$3$ edge $e_2$ would change
the ordinary optimum from weights $4,2,1$ to $4,3,2$, removing
$e_1$. Step~2 therefore rejects this update, and the heavy base
can subsequently be accepted. Thus, unlike contracted greedy,
our algorithm can reject an element even when it improves the
optimum subject to retaining earlier acceptances.

Although this idea avoids the harmful exchange above, if we
update the current configuration only when we accept an element,
the algorithm still fails on simple matroids. For example,
consider a uniform matroid of rank $r$ on $n=4r$ elements with
values in $[1,1+\varepsilon]$, with distinct perturbations to
avoid ties. In the secretary algorithm, this modification keeps
$C=S\cup\ALG$: sample elements are never deleted. Once the
smallest element of $\OPT(C)$ is an accepted element, every
improving insertion would remove it and is forbidden. The
algorithm can then remain permanently stuck with many unused slots.
It accepts only $O(\sqrt r)$ elements in expectation, rather than
$\Omega(r)$.\footnote{Roughly, ignoring the extra restrictions from finalized
samples, after $\sqrt r$ acceptances the smallest accepted
element has only $O(\sqrt r)$ elements below it in the current
greedy solution, on average. Every further acceptance must
replace a sample element below this accepted element, and there
are only $O(\sqrt r)$ such elements left.}
Since all values are nearly $1$, this also gives a vanishing
fraction of optimum value.

\medskip
\noindent\textbf{Idea 2: Update reversibly (Step 2b).}
To overcome this rigidity, the crucial next idea is to update
the current configuration whenever the element passes Step~2,
even when it is not accepted. In the secretary algorithm, this
allows consistent deletions of sample elements. 
Such a deletion
can bring a smaller, not-yet-finalized sample element
 into the optimum, allowing a
later real arrival to replace that sample without removing an
earlier acceptance. This also explains why we retain the sample
configuration rather than only its current optimum: sampled
elements outside that optimum can become relevant after later
updates.

These updates form the heart of the snapshot lemma
(\Cref{lem:unique}): Step~2 compares the same two optima in
either direction, so every permitted update is reversible.
For a fixed processing order, the updates simply permute the
equally likely initial configurations without biasing them.
The configuration faced by an optimum element therefore
remains independent of the processing order, giving it a fair
chance of preceding its possible exchange partner. This is
what turns one-element exchange into the $1/4$ guarantee.

\section{Intersection of Matroids}\label{sec:intersection}

Our algorithm also extends to the intersection of $k$ matroid
constraints. Fix matroids $\calM_j=(V,\calI_j)$ for $j\in[k]$
and nonnegative element values $v_e$. Let $I^\star$ be a
maximum-value common independent set. Write $\OPT_j(w)$ for
greedy in $\calM_j$, and let
$B(w)=\bigcap_{j=1}^k\OPT_j(w)$.
Although $B(w)$ is independent in every matroid, its value need
not approximate that of the optimum common independent set.
We first apply the overlapping-optima preprocessing of
\cite{FSZ-SICOMP22}, losing an
$O(k^2)$ factor. An extension of
\Cref{alg:matroidSec} then loses  another $O(k)$ factor.

\begin{theorem}\label{thm:intersection}
There is an $8k^2(k+1)$-competitive algorithm for secretary under
an intersection of $k$ matroids. It needs only $n=|V|$ and
independence-oracle access to  already-arrived elements.
\end{theorem}

\begin{proof} We begin with the preprocessing of \cite{FSZ-SICOMP22}.

\medskip\noindent \textbf{Preprocessing.}
Let $p=1-1/(2k)$, draw $L\sim\operatorname{Bin}(n,p)$, and
reject the first $L$ elements, forming a sample $P$. Thus $P$
contains each element independently with probability $p$.
Let $H(T)$ be decreasing-value greedy on $T$ subject to all $k$
constraints, using the same fixed tie-breaking order as the
individual greedy algorithms. Define auxiliary values
\[
a_e=
\begin{cases}
v_e,& e\notin P\text{ and }e\in H(P\cup\{e\}),\\
0,&\text{otherwise}.
\end{cases}
\]
The overlapping-optima theorem of
\cite[Theorem 3.3]{FSZ-SICOMP22} gives
\begin{equation}\label{eq:intersection-overlap-preprocess}
\E_P\big[a(B(a))\big]\geq \frac{v(I^\star)}{4k^2}.
\end{equation}
Here $a\leq v$, and $B(a)\subseteq V\setminus P$ because greedy
omits zero-valued elements. Importantly, $a_e$ is computable
when $e$ arrives: the test defining it uses only $P\cup\{e\}$.

\medskip
\noindent\textbf{Algorithm on the remaining elements.}
Let $U=V\setminus P$ and $m=|U|$.  Run \Cref{alg:matroidSec} on the remaining $m$
elements with auxiliary values $a$, fresh randomness, and the
following two changes. Require the test in Step~2 to hold in
every matroid:
\[
\OPT_j(C')\cap F=\OPT_j(C)\cap F\qquad\text{for all }j\in[k].
\]
In Step~2a, accept $e$ exactly when it belongs to every proposed
greedy set. All other steps, including the update in Step~2b and
finalizing every processed element, are unchanged. In this
proof, $\OPT_j(T)$ for $T\subseteq U$ means greedy on $T$ using
$a$. We retain zero-valued elements in the stream and in the
reference-set bookkeeping, but omit them from all greedy scans;
thus they are never accepted. In particular, the input size is
the known number $m$, not the unknown number of positive-valued
elements.

\medskip
\noindent\textbf{Feasibility and the snapshot property.}
As before, the invariant
$\ALG\subseteq\bigcap_{j=1}^k\OPT_j(C)$ proves feasibility:
every permitted update preserves the membership of previous
acceptances, and a new acceptance belongs to every proposed
optimum.

Condition on $P$. The auxiliary vector $a$ is now fixed, and
the remaining physical order is uniform. Let $S$ be the sample
of the new run, and $\Pi$ its virtual processing order.
The interleaving bijection from \Cref{lem:indepPermSet} depends
only on element identities, so for every $T\subseteq U$ and
every permutation $\pi$ of $U$,
\[
\Pr[S=T,\Pi=\pi\mid P]=\frac{1}{2^m m!}.
\]
For fixed $\pi$, the finalized set at each step is a fixed
prefix. Interchanging $C$ and $C\triangle\{e\}$ reverses the
two sides of every equality in Step~2, so its conjunction has
the same outcome in both states. Each reference-set update is
therefore an involution on $2^U$. Consequently, for any fixed
$e\in U$, the reference set $C^e$ immediately before processing
$e$ satisfies
\begin{equation}\label{eq:intersection-snapshot}
\Pr[C^e=T\mid P,\Pi=\pi]\quad =\quad 2^{-m}.
\end{equation}
In particular, conditional on $P$, this snapshot is uniform and
independent of $\Pi$. This argument also applies when some
auxiliary values are zero, because it acts on sets of element
identities rather than on distinct weight configurations.

\medskip
\noindent\textbf{Acceptance probability.}
Fix $e\in B(a)$ and condition further on $C^e=T$ with
$e\notin T$. Only processing $e$ can change its membership in
$C$, so $e\notin T$ means that $e\notin S$ and this is its real
arrival. By \Cref{lem:greedy},
$e\in\OPT_j(T\cup\{e\})$ for every $j$. By
\Cref{lem:heart}, inserting $e$ displaces at most one element
from each greedy set. Let $D$ be the set of distinct displaced
elements; it is determined by $P,T,e$ and satisfies $|D|\leq k$.
The consistency test (Step 2) passes exactly when no element of $D$ has
yet been finalized. By \eqref{eq:intersection-snapshot}, the
conditional order remains uniform, so
\[
\Pr[e\in\ALG\mid P,C^e=T]
~=~\Pr[e\text{ precedes all elements of }D\mid P,C^e=T]
~=~\frac{1}{|D|+1}~\geq~\frac{1}{k+1}.
\]
Since $\Pr[e\notin C^e\mid P]=1/2$, we obtain
$\Pr[e\in\ALG\mid P]\geq1/(2(k+1))$ for every $e\in B(a)$.
Linearity of expectation and
\eqref{eq:intersection-overlap-preprocess} now give
\[
\E[v(\ALG)]~\geq~ \E[a(\ALG)]
~\geq~\frac{\E_P[a(B(a))]}{2(k+1)}
~\geq~\frac{v(I^\star)}{8k^2(k+1)}. \qedhere
\]
\end{proof}

\subsection*{AI Usage}
\Cref{alg:SSProphSec} was found during an extended exploration
with several models available through ChatGPT Pro.
Along the way, we considered several approaches to matroid
secretary, including single-sample
algorithms~\cite{Li-24,DKLRS-FOCS24}, pairing sample and real
elements~\cite{RS-STOC26}, the minimax
principle~\cite{KMPS-FOCS26}, and a reduction to free-order
matroid secretary via tangential
sequences~\cite{FSW-arXiv26}.
The key algorithmic idea emerged recently in this exploration after updating the model to ChatGPT Astra.
We also used ChatGPT Astra and Claude Opus 5 for assistance
with the presentation.
We have checked the mathematical arguments and take
full responsibility for all the content.

\subsection*{Acknowledgments}
We are grateful to Anupam Gupta and S. Matthew Weinberg for discussions on the matroid secretary problem over the
past decade.
We also thank Thomas Kesselheim, Robert Kleinberg, Wenzheng Li, Thiago Oliveira, Kalen Patton, 
 Rebecca Reiffenhäuser, Aviad Rubinstein,  Siddharth M. Sundaram, and Ola Svensson for related discussions.
We thank Anupam Gupta,  S. Matthew Weinberg, and Robert Kleinberg for comments on early versions of this paper.

{\small
\bibliographystyle{alpha}
\bibliography{bib.bib}
}

\end{document}